\documentclass[a4paper, USenglish, cleveref, autoref, thm-restate, pdfa]{lipics-v2021}

\hideLIPIcs%

\title{Min-1-Planarity is NP-Hard} %

\author{Yuto Okada}
{Nagoya University, Japan, JSPS Research Fellow \and \url{https://yutookada.com/en/}}
{research@yutookada.com}
{https://orcid.org/0000-0002-1156-0383}
{}

\authorrunning{Y. Okada} %

\Copyright{Yuto Okada} %



\category{} %

\relatedversion{} %

\funding{Supported by JST SPRING Grant Number JPMJSP2125 and JSPS KAKENHI Grant Number JP26KJ1299.}%

\acknowledgements{I would like to thank Oksana Firman, Marie Diana Sieper, and Alexander Wolff for a discussion about this problem when I visited Universit\"at W\"urzburg in 2023.}%

\nolinenumbers %

\EventEditors{John Q. Open and Joan R. Access}
\EventNoEds{2}
\EventLongTitle{42nd Conference on Very Important Topics (CVIT 2016)}
\EventShortTitle{CVIT 2016}
\EventAcronym{CVIT}
\EventYear{2016}
\EventDate{December 24--27, 2016}
\EventLocation{Little Whinging, United Kingdom}
\EventLogo{}
\SeriesVolume{42}
\ArticleNo{23}

\usepackage{todonotes}

\newcommand{\ThreePartition}{\textsc{3-Partition}}
\newcommand{\MinOnePlanarity}{\textsc{Min-1-Planarity}}

\definecolor{dark blue}{rgb}{0.121,0.47,0.705}
\let\emph\relax\DeclareTextFontCommand{\emph}{\color{dark blue}\em}

\begin{document}

\maketitle

\begin{abstract}
    In this paper, we show that it is NP-hard to determine whether a given graph admits a min-1-planar drawing.
    A drawing of a graph is min-$k$-planar if, for every crossing in the drawing, at least one of the two crossing edges involves at most $k$ crossings.
    This notion of min-$k$-planarity was introduced by Binucci, B{\"{u}}ngener, Di Battista, Didimo, Dujmovi{\'c}, Hong, Kaufmann, Liotta, Morin, and Tappini~[GD 2023; JGAA, 2024] as a generalization of $k$-planarity.
\end{abstract}

\section{Introduction}\label{sec:introduction}

In recent years, the study of beyond-planar graph classes has been one of the most prominent topics in the graph drawing community.
One of the most established such classes is the class of \emph{1-planar graphs}.
A graph is said to be 1-planar if it can be drawn in the plane so that every edge involves at most one crossing.
Such a drawing is called a \emph{1-planar drawing}.
Already in 1965, Ringel~\cite{ringel1965sechsfarbenproblem} studied the chromatic number of 1-planar graphs, and since then there have been many results on 1-planar graphs.
These include, for instance, the NP-hardness of the recognition problem~\cite{DBLP:journals/algorithmica/GrigorievB07,DBLP:conf/gd/KorzhikM08,DBLP:journals/jgt/KorzhikM13}, tight edge-density bound~\cite{DBLP:conf/gd/PachT96,DBLP:journals/combinatorica/PachT97}, and the study of \emph{optimal 1-planar graphs}~\cite{DBLP:journals/algorithmica/Brandenburg18}, namely 1-planar graphs whose edge density matches the upper bound.
The class is naturally extended to \emph{$k$-planar graphs} for $k \geq 1$, where every edge can involve at most $k$ crossings.

This paper concerns a generalization of $k$-planar graphs that is called \emph{min-$k$-planar graphs}.
A graph is said to be min-$k$-planar if it can be drawn in the plane so that, for every crossing between two edges, one of the two edges involves at most $k$ crossings.
Such a drawing is called a \emph{min-$k$-planar drawing}.
A $k$-planar drawing is clearly min-$k$-planar, while the other direction does not hold, as a min-$k$-planar drawing can have an edge involving an unbounded number of crossings.

This notion of min-$k$-planarity was introduced\footnote{It is worth mentioning that, although the terminology has only recently been introduced, Wood and Telle~\cite{DBLP:conf/gd/WoodT06,wt-pdcng-NYJM07} already implicitly studied (circular) min-$k$-planar drawings earlier.} at GD 2023 by Binucci, B{\"{u}}ngener, Di Battista, Didimo, Dujmovi{\'c}, Hong, Kaufmann, Liotta, Morin, and Tappini~\cite{DBLP:conf/gd/BinucciBBDDHKLMT23,DBLP:journals/jgaa/BinucciBBDD00LM24}.
They called edges with more than $k$ crossings \emph{heavy edges} and the remaining edges \emph{light edges}, and observed that heavy edges are pairwise noncrossing and light edges form a $k$-planar drawing.
This structure motivates the study of min-$k$-planar graphs, as allowing heavy edges could lower the minimum $k$ compared to $k$-planar drawings, while benefiting from the low visual complexity of a $k$-planar drawing with pairwise noncrossing heavy edges.

In their paper~\cite{DBLP:conf/gd/WoodT06,wt-pdcng-NYJM07}, they gave several edge-density bounds for min-$k$-planar graphs, including tight upper bounds $4n-8$ and $5n-10$ for $k = 1, 2$, respectively.
They also gave min-$k$-planar graphs that are not $k$-planar for each $k \in \{1, 2\}$, while showing that the classes of optimal min-1-planar graphs and optimal 1-planar graphs coincide.
Hlin{\v{e}}n{\'{y}} and K{\"{o}}dm{\"{o}}n~\cite{DBLP:conf/gd/HlinenyK24} continued the study of min-$k$-planar graphs by investigating the relationship between so-called \emph{simple} and \emph{non-simple} min-$k$-planar drawings.
They showed that for every fixed $k \geq 2$ there is a graph that admits a (non-simple) min-$2$-planar drawing but admits no simple min-$k$-planar drawing.

\subparagraph{Our result.}

We continue the study of min-$k$-planar graphs from the algorithmic point of view.
We consider the problem of recognizing min-$k$-planar graphs, whose complexity, to the best of our knowledge, has been open\footnote{In fact, it was mentioned as an open problem in the GD 2023 presentation of~\cite{DBLP:conf/gd/BinucciBBDDHKLMT23}. The slides are available at \url{https://gd2023.ing.unipg.it/talks/01.1.pdf}. (Accessed on June 30, 2026)}, and show that it is already NP-hard for $k = 1$.

\begin{restatable}{theorem}{ThmMain}\label{thm:main}
    \MinOnePlanarity{} is NP-complete.
\end{restatable}

We remark that the relationship between optimal 1-planarity and optimal min-1-planarity does not carry the NP-hardness of recognizing 1-planar graphs to ours, as optimal 1-planar graphs can in fact be recognized in linear time~\cite{DBLP:journals/algorithmica/Brandenburg18,DBLP:journals/algorithmica/ChenGP06}.

The high-level idea of our reduction for \cref{thm:main} follows that of Grigoriev and Bodlaender~\cite{DBLP:journals/algorithmica/GrigorievB07}, which was used to show the NP-hardness of recognizing 1-planar graphs.
However, as we mentioned above, a min-1-planar drawing allows an edge to have an unbounded number of crossings, which makes it difficult to build gadgets for showing NP-hardness.
Our contribution is to introduce ideas for handling such edges and incorporate them into the reduction of Grigoriev and Bodlaender~\cite{DBLP:journals/algorithmica/GrigorievB07}.
The main ingredient is a tailored uncrossable-edge gadget, which will be given in \cref{sec:uncrossable-edge}.

\subparagraph{Related results.}

M\"unch and Rutter~\cite{DBLP:conf/gd/MunchR24} gave a meta-theoretic algorithm for crossing minimization problems on beyond-planar drawing classes that can be characterized by (finite) forbidden patterns of a certain type.
Their result implies that, for every fixed $k$, testing whether a given graph admits a (simple) min-$k$-planar drawing with at most $c$ crossings is fixed-parameter tractable with respect to $c$.

There are many beyond-planar graph classes that allow an edge with an unbounded number of crossings.
\emph{Fan-planar graphs} are graphs that can be drawn in the plane so that, for every edge $e$, the edges crossing $e$ share the same endpoint.
Their recognition is known to be NP-complete~\cite{DBLP:journals/tcs/BinucciGDMPST15}, even if the rotation system is specified~\cite{DBLP:conf/gd/BekosCGHK14,DBLP:journals/algorithmica/BekosCGHK17}, while the 2-layer variant can be solved in polynomial time~\cite{DBLP:journals/corr/abs-2508-17349}.
\emph{Quasi-planar graphs} are graphs that can be drawn in the plane so that there are no three pairwise crossing edges.
The complexity of their recognition remains one of the long-standing open problems in the community but it is known that the 2-layer variant is already NP-complete~\cite{DBLP:conf/soda/AngeliniLBFP21}.

\section{Preliminaries}\label{sec:preliminaries}

We follow the standard notation in graph theory (see, for example, the textbook by Diestel~\cite{diestel2025graph}).
In this paper we only consider simple graphs, namely, graphs without a self-loop and parallel edges.
The complete graph on $t$ vertices is denoted by $K_t$.
The \emph{length} of a path is defined as the number of edges on the path.

Let $G$ be a graph and let $\Gamma$ be a drawing of $G$ in the plane.
A drawing $\Gamma$ is \emph{simple} if (1) two edges cross at most once, (2) two adjacent edges do not cross, and (3) three edges do not cross at a single point.
For an edge $e \in E(G)$, let $\mathrm{cr}_{\Gamma}(e)$ be the number of edge crossings that $e$ involves.
For an integer $k$, the drawing $\Gamma$ is a \emph{min-$k$-planar} drawing if, for every edge crossing in $\Gamma$ between two edges $e$ and $f$, either $\mathrm{cr}_{\Gamma}(e) \leq k$ or $\mathrm{cr}_{\Gamma}(f) \leq k$ holds.
A graph $G$ is \emph{min-$k$-planar} if it admits a min-$k$-planar drawing.

The problem \emph{\MinOnePlanarity{}} then asks, given a graph $G$, to determine whether $G$ is a min-1-planar graph.
By the following fact, this is in fact equivalent to asking whether $G$ admits a simple min-1-planar drawing.
\begin{lemma}[Proposition 2.2,~\cite{DBLP:conf/gd/HlinenyK24}]\label{lem:simple-min-1-planar}
    Every min-1-planar graph admits a simple min-1-planar drawing.
\end{lemma}

Two drawings are \emph{isomorphic} if there is a homeomorphism of the sphere that maps one drawing to the other.
Two simple drawings are \emph{weakly isomorphic} if they have the same set of pairs of crossing edges.

\section{Uncrossable Edge}\label{sec:uncrossable-edge}

Throughout this paper, by an \emph{uncrossable edge between $u$ and $v$} we mean the subgraph depicted in \cref{fig:uncrossable-edge}.
Formally, for two distinct vertices $u$ and $v$ of a graph, an uncrossable edge between them can be added in the following manner:
\begin{enumerate}
    \item\label{item:add-three-K_4} add three copies of $K_4$ to the graph;
    \item add an edge between $u$ and each of the vertices added in step~\ref{item:add-three-K_4};
    \item\label{item:spokes-u} connect $u$ and each of the vertices added in step~\ref{item:add-three-K_4} with three paths of length 2;
    \item add an edge between $v$ and each of the vertices added in step~\ref{item:add-three-K_4};
    \item\label{item:spokes-v} connect $v$ and each of the vertices added in step~\ref{item:add-three-K_4} with three paths of length 2;
    \item\label{item:add-ten-paths} connect $u$ and $v$ with ten paths of length 2.
\end{enumerate}

\begin{figure}[h]
    \centering
    \includegraphics[page=1]{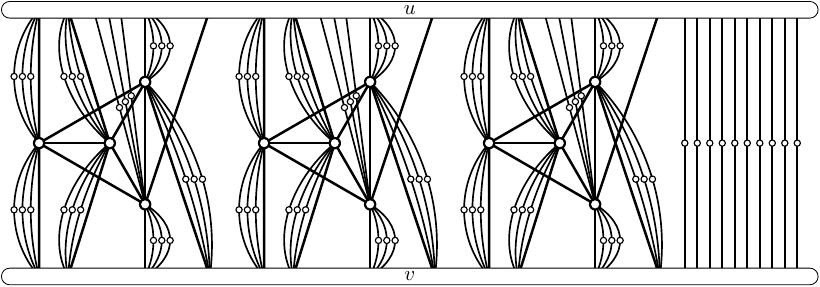}
    \caption{An uncrossable edge between $u$ and $v$ drawn in a min-1-planar manner.}
    \label{fig:uncrossable-edge}
\end{figure}

This section is devoted to showing the following lemma, which will be used in \cref{sec:reduction} to obtain a min-1-planar drawing where each uncrossable edge is essentially crossing-free, namely, not crossed by any other part of the graph.

\begin{lemma}\label{lem:uncrossable-edge}
    Let $G$ be a graph that contains an uncrossable edge between two distinct vertices $u$ and $v$.
    Then, every simple min-1-planar drawing of $G$ contains a curve between $u$ and $v$ that does not cross any edge of $G$ except for the edges of that uncrossable edge.
\end{lemma}

We fix $\Gamma$ as an arbitrary simple min-1-planar drawing of $G$ and show the existence of such a curve in $\Gamma$.
The proof consists of three steps and starts with the following proposition.

\begin{proposition}\label{prop:non-separating-K_4}
    Let $H_1$, $H_2$, and $H_3$ be distinct subgraphs of an uncrossable edge between $u$ and $v$, each of which corresponds to one of the three copies of $K_4$ added in step~\ref{item:add-three-K_4}.
    Then, there exists $H \in \{H_1, H_2, H_3\}$ such that, in the subdrawing of $\Gamma$ induced by $G[\{u,v\} \cup V(H)]$, $u$ and $v$ lie on the boundary of the same face.
\end{proposition}

\begin{proof}
    Let us suppose otherwise for a contradiction.
    Then, each $H \in \{H_1, H_2, H_3\}$ contains in $\Gamma$ a closed curve that separates $u$ and $v$, consisting only of (subcurves of) edges of $H$.
    Namely, there are three edge-disjoint closed curves that separate $u$ and $v$ in $\Gamma$.

    Now consider the ten paths between $u$ and $v$ added in step~\ref{item:add-ten-paths}.
    By the Jordan curve theorem, each path must cross the three closed curves in $\Gamma$, and since that path only has two edges, one of the two edges, say $e$, must cross those closed curves at least twice.
    By the min-1-planarity of $\Gamma$, the edges crossed by $e$ cannot involve any other crossing.
    Hence, there are at least two edges of $H_1$, $H_2$, and $H_3$ that are crossed only by $e$.
    Summing up the numbers of such edges for ten paths, we obtain $|E(H_1) \cup E(H_2) \cup E(H_3)| \geq 20$, contradicting $|E(H_1) \cup E(H_2) \cup E(H_3)| = 3 \cdot \binom{4}{2} = 18$.
\end{proof}

Let $H$ be a subgraph in $\{H_1, H_2, H_3\}$ satisfying the condition in \cref{prop:non-separating-K_4}.
In what follows, we show that there is a desired curve for \cref{lem:uncrossable-edge} between $u$ and $v$ through $H$ and subdivided paths.
To this end, we next enumerate the possible subdrawings $\Gamma'$ of $\Gamma$ induced by $G[\{u,v\} \cup V(H)]$.

\begin{proposition}\label{prop:gammap-two-isomorphism-class}
    Drawing $\Gamma'$ is isomorphic to one of the two drawings depicted in \cref{fig:two-min-1-planar-drawings}.
\end{proposition}

\begin{figure}[h]
    \begin{subfigure}[b]{.49\textwidth}
        \centering
        \includegraphics[page=2]{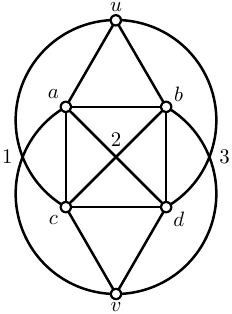}
        \subcaption{}
        \label{fig:two-min-1-planar-drawings:a}
    \end{subfigure}
    \hfill
    \begin{subfigure}[b]{.49\textwidth}
        \centering
        \includegraphics[page=3]{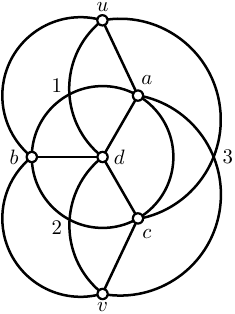}
        \subcaption{}
        \label{fig:two-min-1-planar-drawings:b}
    \end{subfigure}
    \caption{Two min-1-planar drawings of $G[\{u,v\} \cup V(H)]$ with labels for vertices and crossings.}
    \label{fig:two-min-1-planar-drawings}
\end{figure}

\begin{proof}
    As $H$ satisfies the condition in \cref{prop:non-separating-K_4}, it is possible to add a crossing-free curve between $u$ and $v$ to $\Gamma'$.
    This implies that $\Gamma'$ can be extended to a min-1-planar drawing $\Gamma''$ of $K_6$, and in turn, that $\Gamma'$ is a subdrawing of a min-1-planar drawing of $K_6$.
    Since $\Gamma$ is a simple drawing, $\Gamma''$ is also clearly a simple drawing.
    Hence, $\Gamma''$ is weakly isomorphic to one of the 102 simple drawings\footnote{referred to as good drawings in the paper~\cite{Rafla1988}.} of $K_6$ enumerated by Rafla~\cite{Rafla1988}.
    Among them one can verify that only the 77th drawing is min-1-planar, which is weakly isomorphic to the drawing depicted in \cref{fig:min-1-planar-drawing-of-K_6}.

    \begin{figure}[h]
        \centering
        \includegraphics[page=4]{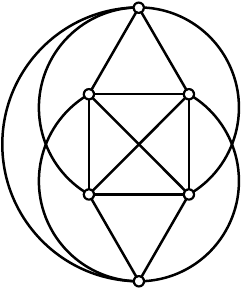}
        \caption{The only possible min-1-planar drawing of $K_6$ up to homeomorphism.}
        \label{fig:min-1-planar-drawing-of-K_6}
    \end{figure}

    Since every edge involves at most one crossing in \cref{fig:min-1-planar-drawing-of-K_6}, more strongly, $\Gamma''$ is isomorphic to this drawing~\cite[Theorem 3.10]{DBLP:journals/dcg/Gioan22}.
    We also remark that one can verify this with a computer program according to the techniques introduced by the authors in~\cite{abrego2015all} to enumerate all the nonisomorphic and non-weakly-isomorphic drawings of a complete graph.

    \begin{figure}[h]
        \begin{subfigure}[b]{.24\textwidth}
            \centering
            \includegraphics[page=5]{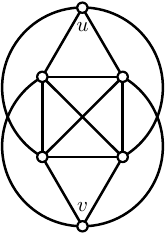}
            \subcaption{}
            \label{fig:four-drawings-of-gammap:a}
        \end{subfigure}
        \hfill
        \begin{subfigure}[b]{.24\textwidth}
            \centering
            \includegraphics[page=6]{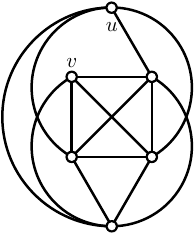}
            \subcaption{}
            \label{fig:four-drawings-of-gammap:b}
        \end{subfigure}
        \hfill
        \begin{subfigure}[b]{.24\textwidth}
            \centering
            \includegraphics[page=7]{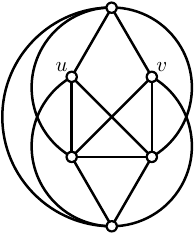}
            \subcaption{}
            \label{fig:four-drawings-of-gammap:c}
        \end{subfigure}
        \hfill
        \begin{subfigure}[b]{.24\textwidth}
            \centering
            \includegraphics[page=8]{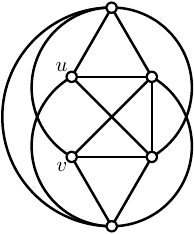}
            \subcaption{}
            \label{fig:four-drawings-of-gammap:d}
        \end{subfigure}
        \caption{Drawing $\Gamma'$ is isomorphic to one of these four drawings.}
        \label{fig:four-drawings-of-gammap}
    \end{figure}

    As edge $\{u, v\}$ is crossing-free in $\Gamma''$, by a symmetric argument $\Gamma'$ is isomorphic to one of the four drawings depicted in \cref{fig:four-drawings-of-gammap}.
    One can confirm that
    \begin{itemize}
        \item the drawings depicted in \cref{fig:four-drawings-of-gammap:a,fig:four-drawings-of-gammap:d} are isomorphic to that of \cref{fig:two-min-1-planar-drawings:a}, and
        \item the drawings depicted in \cref{fig:four-drawings-of-gammap:b,fig:four-drawings-of-gammap:c} are isomorphic to that of \cref{fig:two-min-1-planar-drawings:b},
    \end{itemize}
    which completes the proof.
\end{proof}

Now we are ready to prove \cref{lem:uncrossable-edge}.
We assume for a contradiction that every curve between $u$ and $v$ crosses at least one \emph{external edge}, namely, an edge of $G$ that is not of the uncrossable edge.
This enables us to argue that, for example, the curve $u \rightarrow 1 \rightarrow v$ in \cref{fig:two-min-1-planar-drawings:a} crosses an external edge, and hence, at least one of the edges $\{u, c\}$ and $\{a, v\}$ crosses that external edge.
This then further adds constraints to $\Gamma$: the crossed edge, say $\{u, c\}$, involves at least two crossings and hence the other edge $\{a, v\}$ cannot involve any other crossing.
In this manner, we pick a curve between $u$ and $v$, branch on the choice of edge crossed by an external edge, and show that every case in the end leads to a contradiction.
To show a contradiction, we use the paths of length 2 added in steps~\ref{item:spokes-u}~and~\ref{item:spokes-v}, which we call \emph{wires}.
We do that for each of the two isomorphism classes in \cref{prop:gammap-two-isomorphism-class}, completing the proof for \cref{lem:uncrossable-edge}.

\begin{proposition}\label{prop:gammap-a}
    If $\Gamma'$ is isomorphic to the drawing depicted in \cref{fig:two-min-1-planar-drawings:a}, then $\Gamma$ contains a curve between $u$ and $v$ that does not cross an external edge.
\end{proposition}

\begin{proof}
    For readability, we redisplay \cref{fig:two-min-1-planar-drawings:a} in the following.

    \begin{figure}[h]
        \centering
        \includegraphics[page=2]{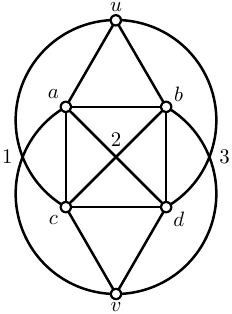}
    \end{figure}

    We assume for a contradiction that every curve between $u$ and $v$ crosses at least one external edge in $\Gamma$.
    Then, as stated above, either $\{u, c\}$ or $\{a, v\}$ crosses an external edge.
    Without loss of generality, we assume that it is $\{u, c\}$.
    We also take the curve $u \rightarrow 3 \rightarrow v$ and branch on the choice between edges $\{u, d\}$ and $\{b, v\}$.

    \subparagraph{If edge $\{u, d\}$ crosses an external edge.}

    Then, edges $\{a, v\}$ and $\{b, v\}$ cannot involve any crossing other than crossings $1$ and $3$.
    Hence, the six wires from $u$ to $c$ or $d$ must cross edge $\{a, b\}$.
    Observe also that, taking the curve $a \rightarrow 2 \rightarrow b$, those six wires cross exactly one of edges $\{a, d\}$ and $\{b, c\}$.
    Without loss of generality, let us assume that it is $\{a, d\}$.
    As $\{a, b\}$ and $\{a, d\}$ involve at least six crossings, a wire between $u$ and $d$ cannot involve any other crossing.
    It is clear that such drawing is impossible and yields a contradiction.

    \subparagraph{If edge $\{b, v\}$ crosses an external edge.}

    Then, edges $\{a, v\}$ and $\{u, d\}$ cannot involve any crossing other than crossings $1$ and $3$.
    This implies that, taking the curve $u \rightarrow 3 \rightarrow d \rightarrow 2 \rightarrow a \rightarrow 1 \rightarrow v$, edge $\{a, d\}$ is crossed by an external edge.
    Hence, also $\{b, c\}$ cannot involve any crossing other than $2$.
    Recalling that $\Gamma$ is a simple drawing, one can observe that a wire between $u$ and $c$ is drawn as one of the two ways depicted in \cref{fig:gammap-a:bv-cross:wires:uc}.
    By symmetry a wire between $b$ and $v$ is also drawn as one of the two ways depicted in \cref{fig:gammap-a:bv-cross:wires:bv}.

    \begin{figure}[h]
        \begin{subfigure}[b]{.49\textwidth}
            \centering
            \includegraphics[page=9]{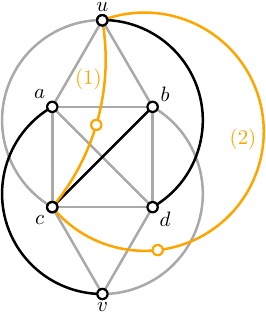}
            \subcaption{}
            \label{fig:gammap-a:bv-cross:wires:uc}
        \end{subfigure}
        \hfill
        \begin{subfigure}[b]{.49\textwidth}
            \centering
            \includegraphics[page=10]{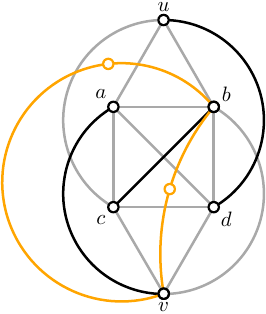}
            \subcaption{}
            \label{fig:gammap-a:bv-cross:wires:bv}
        \end{subfigure}
        \caption{The possible ways of drawing wires (a) between $u$ and $c$ and (b) between $b$ and $v$. Black edges are already crossed by an edge that involves more than one crossing in our assumption and cannot involve any other crossing.}
        \label{fig:gammap-a:bv-cross:wires}
    \end{figure}

    Since there are three wires between $u$ and $c$, there are at least two wires of the same type, either (1) or (2).
    If there are two wires drawn as (1), edges $\{a, b\}$ and $\{a, d\}$ involve at least two crossings, and hence, those wires cannot involve any other crossing.
    The same argument can be applied to (2) with edges $\{b, v\}$ and $\{d, v\}$, which implies that either way there is a wire between $u$ and $c$ that does not cross an external edge.
    By symmetry there is also such a wire between $b$ and $v$.
    Hence, the curve between $u$ and $v$ consisting of such a wire between $u$ and $c$, edge $\{c, b\}$, and such a wire between $b$ and $v$ does not cross an external edge, which yields a contradiction.
\end{proof}

\begin{proposition}
    If $\Gamma'$ is isomorphic to the drawing depicted in \cref{fig:two-min-1-planar-drawings:b}, then $\Gamma$ contains a curve between $u$ and $v$ that does not cross an external edge.
\end{proposition}

\begin{proof}
    For readability, we redisplay \cref{fig:two-min-1-planar-drawings:b} in the following.

    \begin{figure}[h]
        \centering
        \includegraphics[page=3]{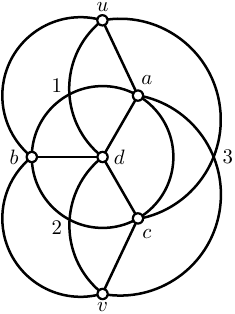}
    \end{figure}

    We assume for a contradiction that every curve between $u$ and $v$ crosses at least one \emph{external edge} in $\Gamma$.
    Similarly to the proof of \cref{prop:gammap-a}, without loss of generality we assume that edge $\{u, c\}$ is crossed by an external edge.
    Hence, edge $\{a, v\}$ cannot involve any crossing other than $3$.
    We again take the curve $u \rightarrow 1 \rightarrow a \rightarrow 3 \rightarrow v$ and branch on the choice between edges $\{u, d\}$ and $\{b, a\}$, as edge $\{a, v\}$ cannot be crossed by an external edge.

    \subparagraph{If edge $\{u, d\}$ crosses an external edge.}

    Then, edge $\{b, a\}$ cannot involve any crossing other than $1$.
    Together with edge $\{a, v\}$ this implies that the six wires from $u$ to $d$ or $c$ must cross edge $\{b, v\}$.
    Now we claim that there is a wire between $u$ and $d$ that crosses edge $\{b, c\}$.
    To see this, suppose otherwise.
    Then, all three wires between $u$ and $d$ must cross three edges $\{a, c\}$, $\{u, c\}$ and $\{c, v\}$.
    Since those three edges involve more than one crossing, the wires must have at least three edges.
    In a similar manner, we can also show the existence of a wire between $u$ and $c$ that crosses edge $\{d, v\}$.
    This leads to a contradiction that both edges of crossing pair $(\{b, c\}, \{d, v\})$ involve more than one crossing.

    \subparagraph{If edge $\{b, a\}$ crosses an external edge.}

    Then, edge $\{u, d\}$ cannot involve any crossing other than $1$.
    Taking the curve $u \rightarrow 1 \rightarrow d \rightarrow 2 \rightarrow v$, this also implies that edge $\{d, v\}$ crosses an external edge and hence edge $\{b, c\}$ cannot involve any crossing other than $2$.
    Recalling that $\Gamma$ is a simple drawing, one can now observe that a wire between $d$ and $v$ is drawn as one of the two ways depicted in \cref{fig:gammap-b:ba-cross:wires:dv}.
    Note that a wire between $d$ and $v$ cannot be drawn in such a way going through the curves between $1$ and $a$, between $u$ and $a$, and between $u$ and $3$.
    As edges $\{b, a\}$ and $\{u, c\}$ already involve a crossing, they would require such a wire to involve at most two crossings.

    \begin{figure}[h]
        \centering
        \includegraphics[page=11]{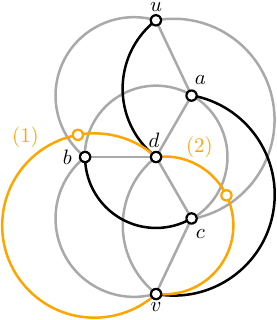}
        \caption{The two possible ways of drawing wires between $d$ and $v$. Black edges are already crossed by an edge that involves more than one crossing in our assumption and cannot involve any other crossing.}
        \label{fig:gammap-b:ba-cross:wires:dv}
    \end{figure}

    As there are three wires between $d$ and $v$, there are at least two wires of the same type, either (1) or (2).
    If there are two wires drawn as (1), due to edges $\{b, a\}$ and $\{u, b\}$ involving more than one crossing, such wires cannot involve a crossing with an external edge.
    The same argument also holds for (2).
    Hence, either way a curve between $u$ and $v$ consisting of edge $\{u, d\}$ and such a wire between $d$ and $v$ does not cross an external edge, which yields a contradiction.
\end{proof}

\section{Reduction}\label{sec:reduction}

In this section, we give a reduction that proves \cref{thm:main} using uncrossable edges in \cref{sec:uncrossable-edge}.
Our reduction basically follows the idea of Grigoriev and Bodlaender~\cite{DBLP:journals/algorithmica/GrigorievB07} to show the NP-hardness of recognizing 1-planar graphs.
However, we make some changes in order to properly handle min-1-planar drawings.

\ThmMain*

\begin{proof}
    The problem is clearly in NP, as one can use (combinatorial information of) a simple min-1-planar drawing as a certificate.
    Observe that the number of crossings in a simple drawing of a graph $G$ is $O(|E(G)|^2)$ as two edges cannot cross twice.

    Our reduction is from \ThreePartition{}.
    Let $n$ be an integer, $X$ be a set of $3n$ distinct positive integers $\{x_1, \dots, x_{3n}\}$, and $T = \left(\sum_{1 \leq i \leq 3n} x_i\right) / n$.
    The problem \ThreePartition{} asks, given $n$ and $X$, to determine if the set $X$ can be partitioned into $n$ triplets such that each of them has a sum of exactly $T$.
    It is well-known that this problem is strongly NP-hard even if $T/4 < x_i < T/2$ holds for every $1 \leq i \leq 3n$~\cite{DBLP:journals/orl/HulettWW08}.
    Hence, we assume the integers in $X$, and therefore $T$, to be polynomial in $n$ and that $T/4 < x_i < T/2$ holds for every $i$.

    \subparagraph{Construction.}

    Let $\langle n, X\rangle$ be an instance of \ThreePartition{} and let $T = \left(\sum_{1 \leq i \leq 3n} x_i\right) / n$.
    From this instance, we construct an instance $\langle G = (V, E) \rangle$ of \MinOnePlanarity{} as follows (see \cref{fig:whole-reduction}), where vertices $a_{n+1}$, $b_{n+1}$, $c_{n+1}$, $d_{i,0}$, and $d_{i,T}$ denote $a_1$, $b_1$, $c_1$, $c_i$, and $c_{i+1}$, respectively:
    \begin{enumerate}
        \item let $V = \{s, t\} \cup \{a_i, b_i, c_i \mid 1 \leq i \leq n\} \cup \{u_i \mid 1 \leq i \leq 3n\} \cup \{d_{i,j} \mid 1 \leq i \leq n, 1 \leq j \leq T-1\}$;
        \item let $E$ $=$ $\{\{s, u_i\} \mid 1 \leq i \leq 3n\}$ $\cup$ $\{\{a_i, a_{i+1}\}, \{b_i, b_{i+1}\}, \{c_i, c_{i+1}\} \mid 1 \leq i \leq n\}$ $\cup$ $\{\{d_{i,j-1}, d_{i,j}\} \mid 1 \leq i \leq n, 1 \leq j \leq T\}$;
        \item for each $1 \leq i \leq 3n$, add $x_i$ paths of length two between $u_i$ and $t$;
        \item\label{item:add-a-path-between-dij-1-dij} for each $1 \leq i \leq n$ and $1 \leq j \leq T$, add a path of length two between $d_{i,j-1}$ and $d_{i,j}$;
        \item for each $1 \leq i \leq n$, add uncrossable edges between $\{s, a_i\}$, $\{a_i, b_i\}$, $\{b_i, c_i\}$, and $\{c_i, t\}$;
        \item for each $1 \leq i \leq n$ and $1 \leq j \leq T-1$, add an uncrossable edge between $\{d_{i,j}, t\}$.
    \end{enumerate}

    \begin{figure}[h]
        \centering
        \includegraphics[page=12]{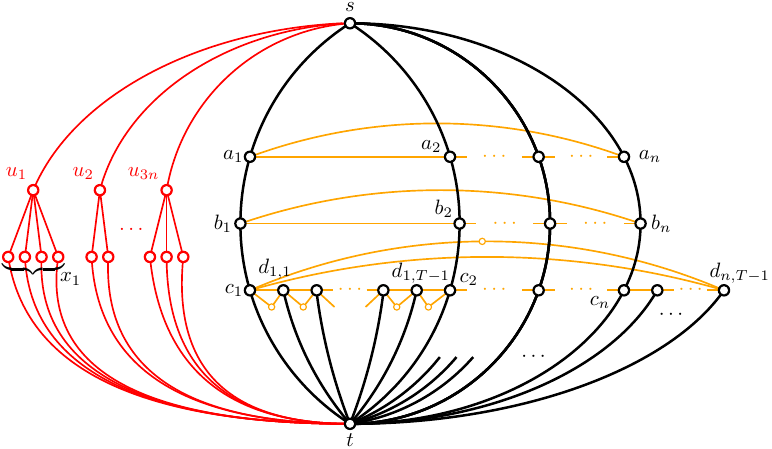}
        \caption{An overview of the graph $G$ we construct. The black edges are uncrossable edges.}
        \label{fig:whole-reduction}
    \end{figure}

    \subparagraph{Completeness.}

    \begin{figure}[h]
        \centering
        \includegraphics[page=13]{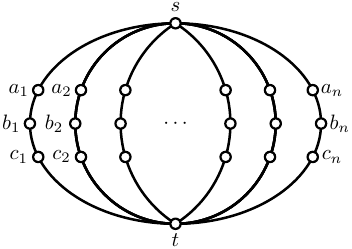}
        \caption{The essentially unique drawing of the $n$ disjoint paths between $s$ and $t$ consisting of uncrossable edges.}
        \label{fig:n-uncrossable-st-path}
    \end{figure}

    Suppose that the set $X$ admits a partition into $n$ triplets, $\mathcal{P} = (P_1, \dots, P_n)$.
    We show that then there is a min-1-planar drawing of $G$.
    Let us first draw the induced subgraph consisting of $s$, $t$, $a_i$'s, $b_i$'s, $c_i$'s, and the uncrossable edges between them as \cref{fig:n-uncrossable-st-path}.
    We draw an uncrossable edge so thinly that we can effectively treat it as an edge.
    Internally we draw it in a min-1-planar manner according to \cref{fig:uncrossable-edge}, contracting vertices $u$ and $v$ in the figure into two points.
    Those $n$ paths make $n$ faces in the drawing, treating the uncrossable edges as edges.
    For each $i$, we draw the following vertices and related edges as \cref{fig:content-inside-a-face} in the face with $a_i$ and $a_{i+1}$ on its boundary:
    \begin{itemize}
        \item if $P_i = \{x_{j_1}, x_{j_2}, x_{j_3}\}$, vertices $u_{j_1}, u_{j_2}, u_{j_3}$ and the middle vertices on the paths between them and $t$;
        \item vertex $d_{i,j}$ for every $1 \leq j \leq T-1$ and the middle vertices added in step~\ref{item:add-a-path-between-dij-1-dij} connecting $d_{i,j-1}$ and $d_{i,j}$ for every $1 \leq j \leq T$.
    \end{itemize}
    \begin{figure}[h]
        \centering
        \includegraphics[page=14]{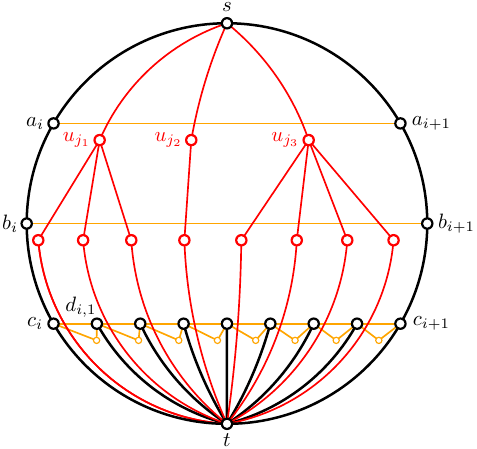}
        \caption{The content to be drawn inside the face with $a_i$ and $a_{i+1}$ on its boundary. This figure corresponds to the case when $P_i = \{x_{j_1} = 3, x_{j_2} = 1, x_{j_3} = 4\}$ and $T = 8$. The black edges are uncrossable edges.}
        \label{fig:content-inside-a-face}
    \end{figure}
    Note that for the case $i = n$ we in fact draw them in the exterior.
    Since $x_{j_1} + x_{j_2} + x_{j_3} = T$, we can assign the $T$ edges between $d_{i,j-1}$ and $d_{i,j}$ for $1 \leq j \leq T$ to the $T$ paths between $t$ and vertices $u_{j_1}, u_{j_2}, u_{j_3}$ bijectively.
    Hence, in this manner we can obtain a min-1-planar drawing of $G$.

    \subparagraph{Soundness.}

    Suppose that graph $G$ admits a min-1-planar drawing.
    We show that then the instance $\langle n, X \rangle$ is a yes-instance.
    To this end, we begin with taking a min-1-planar drawing of $G$ where the uncrossable edges are indeed essentially crossing-free.

    \begin{claim}\label{claim:drawing-s.t.-uncrossable-edges-are-uncrossed}
        There exists a simple min-1-planar drawing of $G$ such that every uncrossable edge does not cross an external edge.
    \end{claim}
    \begin{claimproof}
        Let $\Gamma$ be a simple min-1-planar drawing of $G$, whose existence is guaranteed by \cref{lem:simple-min-1-planar}.
        By \cref{lem:uncrossable-edge}, for every uncrossable edge, say between $u$ and $v$, $\Gamma$ contains a curve between $u$ and $v$ that does not cross an external edge.
        Since the uncrossable edges in $G$ are pairwise edge-disjoint, those curves for uncrossable edges do not cross.
        Hence, we can obtain a desired drawing by redrawing the uncrossable edges along those curves in a min-1-planar manner as \cref{fig:uncrossable-edge}.
    \end{claimproof}

    Let $\Gamma$ be a drawing satisfying the condition of \cref{claim:drawing-s.t.-uncrossable-edges-are-uncrossed}.
    In the following we treat the uncrossable edges just as crossing-free (actual) edges.
    Then, observe that the subdrawing of $\Gamma$ consisting of uncrossable edges $\{s, a_i\}$, $\{a_i, b_i\}$, $\{b_i, c_i\}$, $\{c_i, t\}$ for all $1 \leq i \leq n$ is isomorphic to the drawing depicted in \cref{fig:n-uncrossable-st-path}, since uncrossable edges are crossing-free and there is an edge $\{a_i, a_{i+1}\}$ for every $1 \leq i \leq n$.
    Similarly to the above, each of the $n$ faces made by those uncrossable edges corresponds to each set of a partition of $X$.
    Namely, we define $P_i$ to be the set of integers $x_j$ such that $u_j$ is on the face with vertices $a_i$ and $a_{i+1}$ on its boundary and claim that $\mathcal{P} = \{P_1, \dots, P_n\}$ obtained in this manner is a solution for the original instance $\langle n, X \rangle$.
    In what follows, we show that each face ensures that the corresponding set satisfies the condition for \ThreePartition{}.

    \begin{figure}[h]
        \centering
        \includegraphics[page=15]{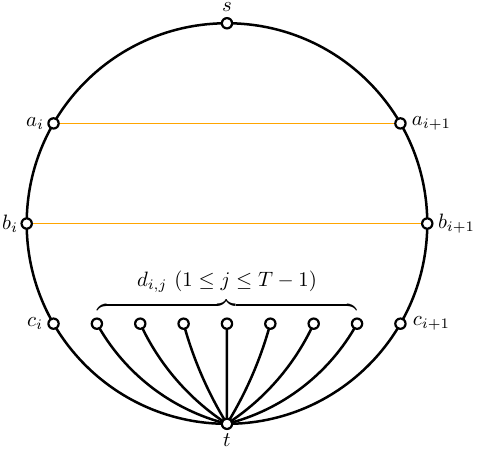}
        \caption{Part of the drawing, inside the face with $a_i$ and $a_{i+1}$ on its boundary, fixed by uncrossable edges (black) and related edges (orange).}
        \label{fig:filters-in-a-face}
    \end{figure}

    Let us first show that a set $P_i$ defined as above has a sum of at most $T$ if $|P_i| \geq 3$.
    This corresponds to the face with $a_i$ and $a_{i+1}$ on its boundary.
    Observe that edges $\{a_i, a_{i+1}\}$, $\{b_i, b_{i+1}\}$, and $\{d_{i, j-1}, d_{i,j}\}$ for $1 \leq j \leq T$ are drawn as \cref{fig:filters-in-a-face} inside the face, due to the simplicity of $\Gamma$ and uncrossable edges.
    In $\Gamma$, there are $\sum_{x \in P_i} x$ paths from $s$ to $t$ drawn in this face, namely of the form $(s, u_j, w, t)$, where $u_j$ satisfies $x_j \in P_i$ and $w$ is one of the vertices connecting $u_j$ and $t$.
    Among them there are at least three edge-disjoint paths since $|P_i| \geq 3$.
    Hence, edges $\{a_i, a_{i+1}\}$ and $\{b_i, b_{i+1}\}$ involve more than one crossing and a path of the form $(s, u_j, w, t)$ must cross them using edges $\{s, u_j\}$ and $\{u_j, w\}$, respectively, implying that those edges cannot involve any further crossing.
    Observe also that each of the $\sum_{x \in P_i} x$ paths crosses edge $\{d_{i, j-1}, d_{i,j}\}$ for at least one $1 \leq j \leq T$, as those edges form a path between $c_i$ and $c_{i+1}$ inside the face.
    Hence, the $\sum_{x \in P_i} x$ paths cross edges $\{d_{i, j-1}, d_{i,j}\}$ for $1 \leq j \leq T$ with their last edges.
    Now it suffices to show that no two of the last edges cross the same edge $\{d_{i, j-1}, d_{i,j}\}$ for some $j$.
    Suppose for a contradiction such two last edges.
    Then, those two edges both cross the other path between $d_{i, j-1}$ and $d_{i,j}$, of length two.
    This is because the first two edges in the paths cannot cross them.
    However, this implies that the two last edges and $\{d_{i, j-1}, d_{i,j}\}$ all involve at least two crossings, which contradicts the min-1-planarity of $\Gamma$.

    Now let us consider a set $P_i$ with $|P_i| < 3$.
    Then, there exists a set $P_j$ with $|P_j| > 3$ since there are $3n$ integers in $X$.
    By the above discussion, $P_j$ has a sum of at most $T$, which is impossible as we assumed $x > T/4$ for every $x \in X$.
    Hence, in fact there is no such $P_i$ and every set $P \in \mathcal{P}$ contains exactly three integers with sum exactly $T$, which completes the proof.
\end{proof}

\section{Conclusions}\label{sec:conclusions}

In this paper, we showed the NP-hardness of \MinOnePlanarity{}.
A natural next open question is the complexity of recognizing min-$k$-planar graphs for fixed $k \geq 2$.
We expect all these problems to be NP-hard, and some ideas in this paper could be used to show this.
However, for every fixed $k \geq 2$, it is known that there is a graph that admits a min-$k$-planar drawing but no simple min-$k$-planar drawing~\cite{DBLP:conf/gd/HlinenyK24}, while our proof heavily uses the fact that every min-1-planar graph admits a simple min-1-planar drawing.
Hence, one may have to come up with a new idea to handle non-simple min-$k$-planar drawings.
It might be also worth considering the simple variants, namely, testing if a graph admits a simple min-$k$-planar drawing.

Another direction would be to revisit the complexity of recognizing quasi-planar graphs.
One can easily observe the following by definition.
\begin{observation}[{\cite[Lemma 3]{DBLP:conf/gd/BinucciBBDDHKLMT23},\cite{DBLP:conf/gd/KlemzKRS21,DBLP:journals/jgaa/KlemzKRS23}}]\label{obs:fan-and-min-1-planar-are-quasi-planar}
    Min-1-planar drawings and simple fan-planar drawings are both quasi-planar drawings.
\end{observation}
Hence, our hardness result can be interpreted as adding another NP-hard recognition problem for a subclass of quasi-planar graphs.
Additionally, in general the following relation holds.
\begin{observation}[{\cite[Lemma 3]{DBLP:conf/gd/BinucciBBDDHKLMT23}}]
    For every $k \geq 1$, min-$k$-planar drawings are $(k+2)$-quasi-planar drawings, namely, there are no $k+2$ pairwise crossing edges.
\end{observation}
Therefore, resolving the complexity of recognizing min-$k$-planar graphs for further $k > 1$ would also help us to understand the structure of ($k$-)quasi-planar graphs better.

\bibliography{main}

\end{document}